\documentclass[11pt]{article}
\usepackage{amsmath,amssymb,amsthm,mathtools,mathrsfs,dsfont}
\usepackage[a4paper,margin=1in]{geometry}
\usepackage{enumitem}
\usepackage{hyperref}
\usepackage{amsrefs,titlesec,fmtcount}
\newtheorem{remark}{Remark}
\newtheorem{lemma}{Lemma}[section]
\newtheorem{theorem}{Theorem}

\newtheorem{corollary}[lemma]{Corollary}
\newtheorem{definition}{Definition}[section]
\newcommand{\define}{\stackrel{\text{\tiny def}}{=}}

\newcommand{\R}{\mathbb{R}}

\newcommand{\id}{\mathds{1}}
\title{Electrostatic inequalities and stability of matter on asymptotically Euclidean manifolds}
\author{David Stuart
 \\{\it{\small Centre for Mathematical Sciences, Wilberforce Road, 
 Cambridge, CB3 OWA,
 England}}\\{\it\small{ email:dmas2@cam.ac.uk}}
}
\date{}

\begin{document}
\maketitle

\thispagestyle{empty}
\vspace{-0.3in}
\begin{abstract}
  We prove stability of nonrelativistic fermionic matter in the presence of fixed nuclei \(\{R_j\}\) on an
  asymptotically Euclidean Riemannian manifold \((M,g)\).
  Upper and lower bounds on sectional curvatures are assumed, but 
  no specific sign conditions on the curvature are imposed.
  The proof follows the route of Lieb-Thirring and electrostatic inequalities, the latter being
  proved by means of a Voronoi decomposition type argument which is however based on the Green functions \(\{G(\,\cdot\,,R_j)\}\) rather than the distance function.
 \end{abstract}
\tableofcontents
  \titleformat{\section}
  {\large\bfseries}
  {\thesection}
  {0.75em}
  {}
\titleformat{\subsection}
  {\normalsize\bfseries}
  {\thesubsection}
  {0.75em}
  {}
\titleformat{\subsubsection}
  {\normalsize\bfseries}
  {\thesubsubsection}
  {0.75em}
  {}
  
\protect\renewcommand{\thesection}{\arabic{section}}
\protect\renewcommand{\thesubsection}{\arabic{section}.\arabic{subsection}}
\protect\renewcommand{\theequation}{\arabic{section}.\arabic{equation}}
\renewcommand*{\theHsection}{\arabic{section}}
\renewcommand*{\theHsubsection}{\arabic{section}.\arabic{subsection}}
\renewcommand*{\theHequation}{\arabic{section}.\arabic{equation}}

\section{Introduction and statement of main theorem}
\setcounter{equation}{0}
The quantum mechanical description of \(N\) identical spinless unit mass fermions on a smooth Riemannian manifold \((M,g)\), in the
presence of \(N\) fixed nuclei located at points \(R_1,\dots R_N\), involves a wave function
\[
\Psi=\Psi(x_1,\dots x_N)\qquad\Psi\in \Lambda^N L^2(d\mu_g)\,,
  \]
taking values in the completed \(N\)-fold antisymmetric product of \(L^2(d\mu_g)\).
Thus \(\Psi\) is a square-integrable completely antisymmetric function of the fermionic coordinates
\(x_1,\dots x_N\,,\) and it evolves according to
  \(
i\hbar\dot\Psi=H\psi
\) where the Hamiltonian is \(H=T+V_{Coul}\) with
\begin{equation}\label{ham}
  T=-\frac{1}{2}\sum_{i=1}^N\,\Delta^i_g\,,\qquad
  V_{Coul}=\sum_{i<j} G(x_i,x_j)-\sum_{i,j}G(x_i,R_j)+\sum_{j<k}G(R_j,R_k)\,.
  \end{equation}
  Here \(\mu_g\), \(\Delta_g\) and \(G\) are, respectively, the volume form, the
  Laplace-Beltrami operator and the Green function determined by the metric \(g\); the index \(i\) indicates that the
  Laplacian acts on the \(i^{th}\) particle. Explicit expressions in local coordinates appear below and in the appendices, together
  with definitions of some notions which are used below. For simplicity
  we have assumed all  particles have identical charge which is absorbed in the Green function. 
We will use as input the following ingredients:

\begin{enumerate}[label=(I\arabic*)]
\item \begin{enumerate}[label=(\alph*)]
\item  {[{\em AE manifold}]}  Numbers \(q>3,\tau>0\) exist such that $(M,g)$ is a smooth three dimensional
  \(W^{1,q}_{-\tau}\)-asymptotically Euclidean manifold in the sense of Definition \ref{bartnikdef}, which has
  positive injectivity radius \(r_{inj}>0\)
  and convexity radius \(r_{conv}\in(0,r_{inj})\).
\item {[{\em Sectional curvature bounds}]} There exist real numbers \(k_1,k_2\) such that for any plane \(\pi\) in any tangent space \(T_xM\) the
  sectional curvature \(k(\pi)\) is bounded above and below by
  \begin{equation}\label{IB1}
-\infty<k_1\leq k(\pi)\leq k_2<+\infty\,,
  \end{equation}
  and there exists
  \(r_{*{\scriptscriptstyle{1}}}\in(0,r_{inj})\) and \(L\geq 1\) depending upon \(k_1,k_2\) such that, using normal geodesic coordinates
  centered at \(y\in M\), 
  \begin{equation}\label{IB2}
L^{-1}\sum_{a=1}^3(\xi^a)^2\leq\sum _{a,b=1}^3 g_{ab}(x)\xi^a\xi^b\leq L\sum_{a=1}^3(\xi^a)^2
\end{equation}
holds for \(d_g(x,y)<r_{*{\scriptscriptstyle{1}}}\), where \(d_g\) is the Riemannian distance function.
Define \(r_*=\min\{r_{conv},
r_{*{\scriptscriptstyle{1}}}\}\).
\item For \(q>3\) as above there exists \(Q<\infty\) such that, with \(x\mapsto r(x)=d_g(x,y)\),
 \begin{equation}\label{IB3}
  \sup_{y\in M}\sum_{|\alpha|\leq 4}\;\int\limits_{r(x)<r_{*}}
  |\partial^{\alpha} g(x)|^q\,\mu_g(dx)\,\leq Q\,,
\end{equation}
in geodesic normal coordinates, where the sum is over all multi-indices of order \(\leq 4\). 
\end{enumerate}
\item {[{\em Non-parabolic}]}  \((M,g)\)
  admits a unique minimal positive Green function $G_y(x)=G(x,y)$ satisfying
  \(\lim_{x\to+\infty}G_y(x)=0\); to be precise \(G\) is a distribution on \(M\times M\)
  which satisfies
  \begin{equation}\label{115}
-\Delta_g G_y=\frac{\delta_y}{\sqrt{\det g}}\,,
    \end{equation}
    and is \(C^\infty\) off the diagonal. We will use the following Hadamard asymptotics
    as \(x\to y\), which is derived in the appendix.
Let \(\beta\in C^\infty(\R)\) be any cut-off function with \(\beta(t)=0\) for
\(t\geq 1\) and \(\beta(t)=1\) for \(t\leq \frac14\), then
  \begin{equation}\label{117}
    G_y(x)=\,\left(\frac{U_0(x,y)}{4\pi r(x)}-\frac{r(x)}{8\pi}U_1(x,y)\right)\beta\left(\frac{r(x)^2}{r_*^2}\right)\,
    +\,{\cal E}
  \end{equation}
 again with  \(r(x)=d_g(x,y)\) and \(U_0,U_1\) smooth functions given in geodesic normal coordinates by \eqref{ste}
and \(\|{\cal E}\|_{C^1(M)}\leq C_H(L,r_*,Q)\) uniformly in \(y\), see the appendix for a derivation.

\item {[{\em Nonnegative Coulomb energy}]} If $\mu=\mu_+-\mu_-$ is a signed Borel measure, its Green potential is defined for \(x\in M\) by
\[
\mathbb{G}\mu(x)=\int_{M} G(x,y)\,\mu(dy)=\int_{M} G(x,y)\,\mu_+(dy)-\int_{M} G(x,y)\,\mu_-(dy),
\]
(under the assumption that one of the terms on the right is finite).
The Coulomb (or Dirichlet) energy
\[
D(\mu,\mu)\define\frac12\int_{M} \mathbb{G}\mu(x)\,\mu(dx)=\frac12\iint_{M\times M}G(x,y)\mu(dx)\mu(dy)
\]
is assumed to be non--negative. (It will always be the case for signed measures \(\mu\) under consideration that if
\(\mu=\mu_+-\mu_-\) is the Hahn decomposition, then
\(D(\mu_+,\mu_+)<\infty\) and \(D(\mu_-,\mu_-)<\infty\,.\))

\item {[{\em Lieb-Thirring}]}  There exists \(C_{LT}>0\) such that
  \begin{equation}\label{LT1}\sum_{j} |E_j|\leq C_{LT}\,\int_{M}|V_{-}|^{\frac{5}{2}}d\mu_g\,,\end{equation}
  where the sum is over the negative eigenvalues of the operator
  \(-\frac12\Delta_g+V,\) and \(V_{-}\) is the
  negative part of \(V\); and for \(N\) fermions with effective one-particle potential
  energy \(V^N_{\mbox{\tiny 1P}}(x_1,\dots x_N)=\sum V(x_i)\) the inequality
  \begin{align}\label{LT2}\begin{split}
    \left(\Psi,\left(T+V^N_{\mbox{\tiny 1P}}\right)\Psi\right)\geq -C_{LT}\,\int_{M}|V_{-}|^{\frac{5}{2}}d\mu_g\,
\end{split}\end{align}
will be assumed to hold for all smooth and normalized \(\Psi\in \Lambda^N(L^2(d\mu_g))\).
\end{enumerate}

\begin{theorem}\label{main}
  Under the hypothese above, there exists a positive number \(C\), depending on
  \(C_{LT},L,r_*,Q\) but independent of the locations \(\{R_j\}\) of the nuclei, such that for
  all smooth normalized \(\Psi\in\Lambda^N L^2(d\mu_g)\)
  \[
\left(\Psi,\,\left(T+V_{Coul}\right)\Psi\right)\geq -CN\,.
    \]
  \end{theorem}
  \noindent
  The proof is presented at the end of \S\ref{electro}. We now make some remarks on the proof and the conditions
(I1)-(I4) imposed above.
  
  \noindent
  {\bf Remarks}
    \begin{enumerate}
    \item[(i)] The only result on stability of matter on a manifold that I am aware of is proved in \cite{MR2787767}
      by a technique which
      generalizes Baxter's potential theoretic approach (from \cite{MR0586803}) to establish electrostatic inequalities on Riemannian manifolds; this uses a nonnegative Ricci curvature condition at various points in the proof.
To avoid curvature restrictions we adopt a different approach here, modeled on that in \cite{MR2583992}*{\S 5 and \S 7.2} and
      \cite{solorev}, with some changes, most significantly:
    \begin{enumerate}\item[(a)] Harmonic measure on level sets of the Green function is used in place of the standard Eucldean \(\R^3\) Newton shell theorem to smear out the charge with control of the interaction energy; this was already exploited in \cite{MR2787767}.
    \item[(b)] Hadamard expansions for the Green function are used to control the level sets of the Green function and the regularized self-interaction \eqref{selfi}.
    \item[(c)] The procedure based on the Voronoi decomposition used in \cite{MR2583992}*{Chapter 5} to bound the Coulomb interaction from below by a sum
      of one body potentials is modified to avoid the need to consider interfaces between the Voronoi cells, which appears challenging in the manifold case.
      \end{enumerate}
\item[(ii)] One could also imagine applying the approach here to the case of a nonnegative Ricci curvature complete Riemannian manifold studied in \cite{MR2787767} - it seems quite possible it would work it, but would likely give worse constants than achieved in \cite{MR2787767} due to the way the regularization here is constrained by the use of local Hadamard expansions for the Green function.
 \item[(iii)]
    The condition \eqref{IB2} is actally a consequence of \eqref{IB1} and the
    Rauch comparison principle \cite{MR1393941}*{\S VIII.4}.
    \item[(iv)] The number \(Q\) appearing in \eqref{IB3} could
    potentially be controlled quantitatively in terms of norms for the Ricci curvature, using harmonic coordinates
    as in \cite{MR0644518}. Harmonic coordinates will always exist locally, but to work with them quantitatively
    a lower bound on the harmonic radius is needed. ( The {\em harmonic radius} \(r_{harm}\) of \(M\) is the
    largest positive number such that any ball \(B_y(r)\) of radius \(r<r_{harm}\) is a subset of a local
    harmonic coordinate chart.) In the present asymptotically Euclidean case it seems that a combination of
    \cite{bartnik}*{Proposition 3.3} and \cite{MR1074481}*{Lemma 2.2} would provide an atlas of harmonic charts
    which would yield quantitative control in terms of \(k_1,k_2\), the injectivity radius and the diameter of
    the compact core \(K\subset M\) as well as appropriate \(W^{k,q}\) norms of the Ricci tensor.
  \item[(v)]
    There is a discussion of the derivation of Lieb-Thirring inequalities on Riemannian manifolds at the beginning of Section 5 in \cite{MR2787767} which
    essentially carries over to the class of asymptotically Euclidean manifolds considered here. The argument given there reduces the issue to showing that
    the volume of the sets \eqref{defom} is \(\lessapprox \alpha^{-3}\), which follows from the Hadamard asymptotics \eqref{117}.
  \end{enumerate}
  \section{Harmonic Measure and Newton interaction inequality}
  \setcounter{equation}{0}
\subsection{Green Functions and harmonic measure}

Let $(M,g)$ be a smooth oriented Riemannian manifold.
In local coordinates the  metric is given by
\(
g = g_{ij}\,dx^i dx^j,
\)
and the corresponding Riemannian volume form is
\(
d\mu_g
=
\sqrt{\det g}\,dx^1dx^2dx^3.
\) The induced metric on the cotangent space is represented by the inverse matrix \(g^{-1}\), whose entries
will be written \(g^{ij}\) as usual.
With these definitions
the Laplace-Beltrami operator is
\[
\Delta_g u
=\frac{1}{
\sqrt{\det g}}\,
\partial_i
\left(
\sqrt{\det g}\,g^{ij}\partial_j u
\right).
\]

Let $G(x,y)$ denote a Green function for $-\Delta_g$, namely a distributional solution of \eqref{115},
\(G\), which is smooth on the complement of the diagonal in \(M\times M\), and
\(x\mapsto G(x,y)\) is locally integrable; an
equivalent version of the above is to require that, for every smooth compactly supported test function $f$,
\[
\int_M -\Delta_g  f(x)G(x,y)\,d\mu_g(x)
=
f(y).
\]

Let $\Omega\subset M$ be an open bounded domain having \(C^1\)  boundary \(\partial\Omega\)
with outward-pointing unit co-normal vector field $\nu$.
Then the Green identity for a pair of \(C^2\) functions \(U,V\) reads
\[
\int_\Omega
(U\Delta_g V - V\Delta_g U)\,d\mu_g
=
\int_{\partial\Omega}
\left(
U\langle \nu,\nabla V\rangle
-
V\langle \nu,\nabla U\rangle
\right)d\sigma_g.
\]
were
\[
\langle \nu,\nabla V\rangle
=
\nu_i g^{ij}\partial_j V
\]
is the normal derivative on the boundary, and \(d\sigma_g\) is the area form defined
by the induced metric, see appendix for formulae in local coordinates.

Fix $y\in\Omega$.
The {\em harmonic measure} associated to $y\in\Omega$ is a measure supported on
the boundary, given by the formula
\[
d\gamma_{y}(x)
=
-\langle \nu,\nabla_x G(x,y)\rangle d\sigma_g(x)
\qquad x\in \partial\Omega.
\]
Applying the Green identity with $U=1$ and $V=G(\cdot,y)$ gives the normalization
\[
\int_{\partial\Omega} d\gamma_{y} =1.
\]
Under suitable assumptions the measure $\gamma_{\alpha,y}$ is non-negative. In particular, inside a
sufficiently small geodesic ball
the Hadamard parametrix expansion \eqref{117} is valid in normal geodesic coordinates \eqref{gnc}, allowing more detailed investigation.
 We will actually use 
regions of the form
\begin{equation}\label{defom}
\Omega_{\alpha,y}
=
\{z\in M:G(z,y)>\alpha\}\cup\{y\},
\end{equation}
which for large \(\alpha\) is an open neighbourhood of \(y\) with smooth boundary \(\partial\Omega_{\alpha,y}\), which can
be written as a graph over the unit sphere, as we now explain (with more details in the appendix).
In a neighbourhood \(B_0(r_{*})=\{x:d_g(x,y)<r_{*}\}\) of \(y\), which will be taken as the origin,
there exists a system of geodesic spherical coordinates
\((r,\{\theta^\alpha\}_{\alpha\in \{1,2\}})\)
in which the metric takes the form
\(
g=dr^2+g_{\alpha\beta}(r,\theta)d\theta^\alpha d\theta^\beta\,,
\)
with \(g_{\alpha\beta}\) bounded above and below in terms of the sectional curvature bounds in \eqref{IB1}.
Further choosing \(\rho_1=r_*/2\) in the Hadamard asymptotics derived in the appendix, we
have \(\beta\equiv 1\) for \(r<r_*/2\) and writing
\(\hat G(r,\theta^1,\theta^2)=G_y(x(r,\theta^1,\theta^2))\) for the Green function written in the
spherical geodesic coordinates, and similarly for other functions, the Hadamard expansion takes the form
\[
\hat G=\frac{\hat U_0}{4\pi r}-\frac{1}{8\pi}r\hat U_1+\hat{\cal E}
\]
where $U_0,U_1$ are given by \eqref{ste} and \(\|\hat{\cal E}\|_{C^1(B_0(r_{*}/2))}
\leq const. (L,r_*,Q)\).
By \eqref{bmieq} and the Sobolev inequality, \(U_0\)
is bounded in \(C^3(B_0(r_{*}))\), and \(U_1\) is bounded in \(C^1(B_0(r_{*}))\) with the
same dependencies, and so
\[
\frac{\partial\hat G}{\partial r}=-\frac{1}{4\pi r^2}\left(1+r\hat\Gamma\right)
\,,\;\hbox{where}\;\|\hat{\Gamma}\|_{C(B_0(r_{*}/2))}
\leq const. (L,r_*,Q)\,.
\]
It follows from the inverse function theorem that the level set condition
\(\hat G=\alpha\) can be solved locally for a smooth function
\begin{equation}\label{displays}
r=R(\alpha,\theta^1,\theta^2)\,,\quad\hbox{such that}\quad \hat G\left (R(\alpha,\theta^1,\theta^2),\theta^1,\theta^2\right)=\alpha\,,
\end{equation}
with \(R=(4\pi\alpha)^{-1}(1+o(1))\) as \(\alpha\to+\infty\).
For fixed \(\theta\) the function \(R\)
is monotonically decreasing for large \(\alpha\) since \(\hat G\) is monotonically decreasing for \(r\)
sufficiently small that \(|r\hat\Gamma|<1\,.\)
Now
\eqref{displays} displays
the boundary \(\partial\Omega_{\alpha,y}\) as a \(C^\infty\) graph over the sphere
for \(\alpha>\alpha_*\).
In this case we write
\(\gamma_{\alpha,y}\) for the corresponding harmonic measure, which is a {\em nonnegative} normalized measure, absolutely continous with respect to the induced surface measure (and
by the Hopf lemma, see \cite{MR1625845}, it will even have strictly positive density as long as the boundary is
\(C^2\)).

\subsection{Harmonic measure and the curved space Newton shell theorem}

Suppose $U\in C^2(M)$ is harmonic on $\Omega_{\alpha,y}$, i.e.
\(
\Delta_g U =0.
\)
The Green identity with $V=G(\cdot,y)-\alpha$ yields (since the boundary is \(C^\infty\))
\[
U(y)
=
-\int_{\partial\Omega_{\alpha,y}
  }
U(x)
\langle \nu,\nabla_x G(x,y)\rangle
\,d\sigma_g(x)=
\int_{\partial\Omega_{\alpha,y}
  }U(x)\,d\gamma_{\alpha,y}(x).
\]
Thus harmonic measure reproduces harmonic functions from boundary values on \(\Omega_{\alpha,y}\).
More generally,
\[
U(y)
=
-\int_{\Omega_{\alpha,y}
  }\left(G(x,y)-\alpha\right)\Delta_g U(x)d\mu_g(x)
+
\int_{\partial\Omega_{\alpha,y}
  }U(x)\,d\gamma_{\alpha,y}(x).
\]

In particular, if $\Omega_{\alpha,y}$ is as above,
then for exterior points \(z\in \Omega_{\alpha,y}^c\) one obtains by choosing
\(U=G(\cdot,z)\), the representation
\begin{equation}\label{ext}
G(y,z)
=
\int_{\partial\Omega_{\alpha,y}} G(x,z)\,d\gamma_{\alpha,y}(x).
\end{equation}


Now suppose that \(z\) is an interior point, i.e.,
\(
z\in \Omega_{\alpha,y}.
\)
Since the Green function assumes the constant boundary value $\alpha$ on the level set
\[
\partial\Omega_{\alpha,y}
=
\{w:G(w,y)=\alpha\},
\]
we obtain this time, after inserting \(U=G(\cdot,z)\),
\begin{align}\notag G(y,z)&=
\int_{\Omega_{\alpha,y}}
\left(G(x,y)-\alpha)\right)\frac{\delta(x-z)}{\sqrt{\det g}}\,d\mu_g(x)
+  
\int_{\partial\Omega_{\alpha,y}}
G(x,z)
\,d\gamma_{\alpha,y}(x)\\
&=G(z,y)-
\alpha+
\int_{\partial\Omega_{\alpha,y}}
G(x,z)
d\gamma_{\alpha,y}(x).
                            \notag\end{align}
Rearranging and using \(G(x,y)=G(y,x)\) we find, for interior \(z\),
\begin{equation}\label{int}
\int_{\partial\Omega_{\alpha,y}}
G(x,z)
d\gamma_{\alpha,y}(x)=\alpha\,,
\end{equation}
and the integral on the left defines a function which is
continuous across the boundary. (As an integral it has a similar singularity structure to a single layer potential, see 
\cite{MR1681462}*{\S 3D}.)

All together, by \eqref{ext} and \eqref{int},  the potential generated by the harmonic measure is constant throughout the interior region
\(\Omega_{\alpha,y}\), equal to the
constant value \(\alpha\) defining the level set, while outside it is the potential produced by a point charge at
\(y\in \Omega_{\alpha,y}\). This is the curved-space analogue of the classical theorem:
%
%
%
in Euclidean space one has the Newtonian kernel
\[
G(x,y)=\frac{1}{4\pi |x-y|}
\qquad (n=3).
\]
and the level sets \(\partial\Omega_{\alpha,y}\) are just spheres centred at \(y\), so the above conclusions reproduce the
usual results on the gravitational potential outside and inside a uniform spherical shell.



\subsection{Interaction estimates} Observing that \eqref{ext} and \eqref{int} imply that  
\begin{equation}\label{nti}
G(y,z)\geq \int_{\partial\Omega_{\alpha,y}}
G(x,z)
d\gamma_{\alpha,y}(x)
\end{equation}
holds always, and applying it twice
\begin{equation}\label{ntii}
G(y,z)\geq \int_{\partial\Omega_{\alpha,y}}
G(x,z)
d\gamma_{\alpha,y}(x)\geq
\iint_{\partial\Omega_{\alpha,y}\times \partial\Omega_{\alpha',z} } G(x,w) d\gamma_{\alpha,y}(x)d\gamma_{\alpha',z}(w)
  \end{equation}
  as long as \(\alpha,\alpha'\) are sufficiently large real numbers. The interpretation of this is that if point charges
  are smeared out on sufficiently small ``harmonic measures'' whose supports surround them, the interaction energy can only decrease or remain unchanged, again as in the classical case.

  Also from \eqref{int} a further integration gives the value of the regularized self-interaction as
  \begin{equation}\label{selfi}
D(\gamma_{\alpha,y},\gamma_{\alpha,y})=\frac12\alpha\,.
    \end{equation}

\section{An Electrostatic Inequality}\label{electro}
\setcounter{equation}{0}
\begin{lemma}[Electrostatic inequality]
\label{prop:electrostatic}
Define
\begin{equation}
\Phi\;\define\;\sum_k G_k-\max_k G_k\;=\;\min_{1\le l\le N} K_l,\qquad\hbox{where}\quad
K_l\define\sum_{k\neq l}G_k .
\label{32p}
\end{equation}
Let $\nu=-(\Delta_g\Phi)\,d\mu_g$. Then $\nu\ge0$, $\;\Phi=\mathbb{G}\nu$, $\;W-\Phi=\max_k G(\cdot,R_k)$,
and $D(\nu,\nu)\le\sum_{j<k}G(R_j,R_k)$. Consequently, for every signed measure $\mu=\mu_+-\mu_-$
with $D(\mu_\pm,\mu_\pm)<\infty$,
\begin{equation}
D(\mu,\mu)-\int_M\Phi\,d\mu+\sum_{j<k}G(R_j,R_k)\;\ge\;0 .
\label{35p}
\end{equation}
\end{lemma}

\begin{proof}
Each $K_l$ is a finite sum of Green functions, hence nonnegative and superharmonic on all of $M$;
and a minimum of finitely many superharmonic functions is superharmonic, since a harmonic $h$ with
$h\le\Phi$ on $\partial B$ satisfies $h\le K_l$ in $B$ for every $l$. All these functions are also
lower semi-continuous, and the Riesz decomposition applies:
\begin{quote}{\em Any superharmonic and lower semicontinuous $u\geq 0$ with
\[
\mathbb{G}\nu<\infty\,,\qquad\hbox{where}\quad\nu=-\Delta_g u\quad\hbox{(distributionally),}
  \]
can be written
\(
u=h+\mathbb{G}\nu,
\)
where $h$ is the greatest harmonic minorant of $u$.}
\end{quote}
Since
$0\le\Phi\le W$ and $W\to0$ at infinity, the maximum principle squeezes every harmonic minorant of
$\Phi$ to zero, so the Riesz decomposition gives $\Phi=G\nu$.
Superharmonicity implies $\nu\ge0$.
The identity $W-\Phi=\max_kG_k$ is the
definition of $\Phi$.

Near $R_l$ the function $G_l$ occurs in every $K_m$ with $m\neq l$ and blows up, while $K_l$ stays
bounded; so $\Phi=K_l$ on a neighbourhood of $R_l$, and in particular $\Phi$ is smooth there with
$\Phi(R_l)=\sum_{k\neq l}G(R_l,R_k)$. This is exactly where (3.2) fails: there the divergent
variable survives in the limit, here it is absent from the branch that attains the minimum. Hence,
using $\Phi\le W$ and then $G(x,R_j)=G(R_j,x)$,
\begin{align*}
D(\nu,\nu)=\tfrac12\int_M\Phi\,d\nu\;&\le\;\tfrac12\int_M W\,d\nu
=\tfrac12\sum_j\int_M G(x,R_j)\,\nu(dx)\\
&=\tfrac12\sum_j\Phi(R_j)=\sum_{j<k}G(R_j,R_k)<\infty .
\end{align*}
Finally $2D(\mu,\nu)=\int G\nu\,d\mu=\int\Phi\,d\mu$, so by (I3)
\[
0\;\le\;D(\mu-\nu,\mu-\nu)=D(\mu,\mu)-\int_M\Phi\,d\mu+D(\nu,\nu),
\]
which with the previous display is \eqref{35p}.
\end{proof}

\begin{corollary}\label{315}
Choose $\alpha_i>0$ with $d_g(R_j,\Omega_{\alpha_i,x_i})>0$ for all $i,j$, and set
$\gamma_i=\gamma_{\alpha_i,x_i}$. Then
\begin{multline}\label{electroeq}
\tfrac12\sum_{i\neq j}G(x_i,x_j)-\sum_{i,j}G(x_i,R_j)+\tfrac12\sum_{j\neq k}G(R_j,R_k)\\
\ge\;-\sum_i D(\gamma_i,\gamma_i)-\sum_i\int\max_k G(x,R_k)\,\gamma_i(dx).
\end{multline}
\end{corollary}

\begin{proof}
Take $\mu=\sum_i\gamma_i$, which is nonnegative, so $D(\mu,\mu)<\infty$ and the Hahn decomposition
is trivial. By (2.6),
\[
D(\mu,\mu)=\sum_{i<j}\iint G\,d\gamma_i\,d\gamma_j+\sum_i D(\gamma_i,\gamma_i)
\;\le\;\sum_{i<j}G(x_i,x_j)+\sum_i D(\gamma_i,\gamma_i);
\]
and since every $R_j$ lies outside $\operatorname{supp}\gamma_i$, (2.3) gives
$\int W\,d\gamma_i=\sum_j G(x_i,R_j)$, whence
\[
\int_M\Phi\,d\mu=\sum_{i,j}G(x_i,R_j)-\sum_i\int\max_k G(x,R_k)\,\gamma_i(dx).
\]
Substituting both into \eqref{35p} and rearranging gives the claim.
\end{proof}
    \paragraph{Proof of Theorem \ref{main}}
    To make use of Corollary \ref{315}  we employ the Hadamard expansions to
    show that the final term is bounded below by a sum of one particle
    potentials, so as to argue as in \cite{solorev}.
    Define, for each \(i\), the distance of the \(i^{th}\) fermion to the nearest nucleus
    \[\delta_i=\min\{d_g(x_i,R_j):j=1,\dots N\}\,,\]
    and then let
    \begin{equation}\label{rhoi}
      \rho_i<\frac12\min\left\{\delta_i,r_*
        \right\}
      \end{equation}
Now choose \(\alpha_i>\alpha_*\) sufficiently large that
\begin{equation}\label{choo}
\overline{\Omega_{\alpha_i,x_i}}\subset B_{x_i}(\rho_i)\,.
\end{equation}
It will be necessary to be more precise, which we do with consideration to the Hadamard approximation.
By \eqref{117} we have,
\[
  G(x,R_j)-\frac{1}{4\pi d_g(x,R_j)}=\frac{\beta(d_g(x,R_j)^2)U_0-1}{4\pi d_g(x,R_j)}-
  \frac{1}{8\pi}\beta(d_g(x,R_j)^2) U_1d_g(x,R_j)+\mathcal{E}
\]
where \(U_0,U_1\) are the functions given in \eqref{ste} (in normal coordinates with  \(R_j\) taken as origin).
This gives
\begin{equation}\label{boundG}
\sup_{x\in M}\,\left|  G(x,R_j)-\frac{1}{4\pi d_g(x,R_j)}\right|\leq \Lambda_0(r_*,L,Q)\,,
\end{equation}
and so if \(x\in\hbox{supp}\gamma_i\) the triangle inequality gives \(d_g(x,R_j)\geq d_g(x_i,R_j)-\rho_i\),
which is positive for sufficiently small \(\rho_i\), and hence
\[
G(x,R_j)\leq \frac{1}{4\pi (d_g(x_i,R_j)-\rho_i)}+\Lambda_0(r_*,L,Q)\,.
  \]
  As a consequence, for \(x_i\in S_f\) where \(S_f=\{z:\min_{j} d_g(z,R_j)\geq r_*\}\) we have
\(d_g(x_i,R_j)-\rho_i\geq r_*/2\)
  , {\em for all } \(j\), and hence
  \[
G(x,R_j)\leq \frac{1}{2\pi r_*}+\Lambda_0(r_*,L,Q)=\Lambda_1(r_*,L,Q)\,.
\]
So for \(i\) such that \(x_i\in S_f\), there holds
    \(  \int \max_{j}G(x,R_j)\gamma_i(dx)\leq \Lambda_1(r_*,L,Q)\,.\) 
    Although this no longer holds for \(x_i\in S_f^c\), we can instead exploit the definition of \(\delta_i\) to deduce
    \(d_g(x_i,R_j)\geq\delta_i>0\) for each \(j\), and hence
    \(x\in\hbox{supp}\gamma_i\implies d_g(x,x_i)\leq\rho_i<\delta_i/2\) which in turn implies \(d_g(x,R_j)\geq \delta_i/2\), and hence
\[  \max_{j}G(x,R_j)\leq \frac{1}{2\pi \delta_i}+\Lambda_0(r_*,L,Q)\,.\] 
All together then
\begin{equation}\label{atn}
  -\int\,\max_{j}G(x,R_j)\gamma_i(dx)\geq\begin{cases}& -\Lambda_1(r_*,L,Q)\phantom{\frac{1}{2\pi \min_{j} d_g(x_i,R_j)}+\Lambda_0}\;\quad\quad\hbox{for}\;x_i\in S_f \\
                                                      & -\left(\frac{1}{2\pi \min_{j} d_g(x_i,R_j)}+\Lambda_0(r_*,L,Q)\right)\qquad\hbox{for}\;x_i\in\, S_f^c.
                                                        \end{cases}
  \end{equation}
and \(S_f^c=\bigcup_{j} B_{R_j}(r_*)\). Therefore we have the lower bound
  \begin{align*}\frac12\sum_{i\neq j} G(x_i,x_j)-\sum_{i,j}G(x_i,R_j)+\frac12\sum_{k\neq k}G(R_j,R_k)
&\geq     -\frac12\sum_{i}D(\gamma_i,\gamma_i) -\sum\limits_{i} \id_{S_f}(x_i)\,\Lambda_1\\
&\qquad -\sum\limits_{i}\id_{S_f^c}(x_i)\left(\Lambda_0(r_*,L,Q)+v(x_i)\right)
  \end{align*}
where
\[
  v(z)=\frac{1}{2\pi\min_{j} d_g(z,R_j)}\,.
\]
It remains to consider the self-interaction term \(D(\gamma_i,\gamma_i)\) which was computed in \eqref{selfi}.
As in \eqref{boundG} we have, by \eqref{117} with
\(y={x_i}\),
\begin{equation}\label{boundGi}
\frac{1}{4\pi d_g(x,x_i)}-\Lambda_0\leq G(x,x_i)\leq \frac{1}{4\pi d_g(x,x_i)}+\Lambda_0\,.
\end{equation}
Recall \eqref{rhoi}, and choose
\[
\alpha_i=\Lambda_0+\frac{1}{2\pi\delta_i}+\frac{1}{2\pi r_*}
\]
then \(G(x,x_i)>\alpha_i\) implies \(d_g(x,x_i)<\delta_i r_*/(2(r_*+\delta_i))<(r_*\wedge \delta_i)/2\) which guarantees \eqref{choo}.
By \eqref{selfi} we have
\(D(\gamma_i,\gamma_i)=\alpha_i/2\) and so conclude
that \eqref{electroeq} is bounded below by
\begin{align*}
  \frac12\sum_{i\neq j} G(x_i,x_j)-\sum_{i,j}G(x_i,R_j)+\frac12\sum_{k\neq k}G(R_j,R_k)
&\geq    -\sum\limits_{i}\frac12\left(\Lambda_0+\frac{1}{2\pi\delta_i}+\frac{1}{2\pi r_*}\right)\\
&\quad  -N\,(\Lambda_0+\Lambda_1)-\sum\limits_{i}\,v(x_i)\id_{S_f^c}(x_i)\,.
\end{align*}
Writing \(\delta_i^{-1}+r_*^{-1}\leq 2r_*^{-1}\id_{S_f}+2\delta_i^{-1}\id_{S_f^c}\) we obtain
\begin{align*}
  \frac12\sum_{i\neq j} G(x_i,x_j)-\sum_{i,j}G(x_i,R_j)+\frac12\sum_{k\neq k}G(R_j,R_k)
&\geq   
-N\,\left(\frac32\Lambda_0+\Lambda_1+\frac{1}{2\pi r_*}\right)-\sum\limits_{i}3\,v(x_i)\id_{S_f^c}(x_i)\,.
\end{align*}
Now since \(S_f^c=\bigcup_{j} B_{R_j}(r_*)\) we may estimate this integral using spherical geodesic coordinates
centered at \(R_j\) in each ball. So put \(r=d_g(z,R_j)\) and use \eqref{IB2} to bound
\(\mu_g(dz)\leq L^{\frac32}dz\), and hence
\begin{align*}
  \int_{M} &|v(z)\id_{S_f^c}(z)|^p\mu_g(dz)\leq\sum_{j}\int_{B_{R_j}(r_*)}(2\pi d_g(z,R_j))^{-p}\mu_g(dz)\\
  & \leq 4\pi N L^{3/2} \int_{r<r_*}(2\pi r)^{-p}r^2dr\,=8N\pi(2\pi)^{-5/2} L^{3/2}r_*^{1/2}\qquad (p=5/2)\, .
  \end{align*}
  The following lower bound then follows from \eqref{LT2}:
  \[
\left(\Psi,\,\left(T+V_{Coul}\right)\Psi\right)
\geq -N\,\left(\frac32\Lambda_0+\Lambda_1+\frac{1}{2\pi r_*}\right)
-C_{LT}3^{5/2}8N\pi(2\pi)^{-5/2} L^{3/2}r_*^{1/2}\,,
\]
  and hence proof of the theorem.

  \protect\renewcommand{\thesection}{\Alph{section}}
\protect\renewcommand{\thesubsection}{\Alph{section}.\arabic{subsection}}
\protect\renewcommand{\theequation}{\Alph{section}.\arabic{equation}}
\renewcommand*{\theHsection}{\Alph{section}}
\renewcommand*{\theHsubsection}{\Alph{section}.\arabic{subsection}}
\renewcommand*{\theHequation}{\Alph{section}.\arabic{equation}}

\setcounter{section}{0}
\setcounter{subsubsection}{0}
\setcounter{equation}{0}

\section{Appendices}
\setcounter{equation}{0}
\subsection{Induced Riemannian Volume Form and Divergence Theorem}

Let $(M^n,g)$ be an oriented Riemannian manifold,
covered by a \(C^\infty\) atlas of charts  
\(\{\kappa,{\mathcal{O}}_\kappa,\tilde {\mathcal{O}}_\kappa\} \), where
\(\kappa:{\mathcal{O}}_\kappa\to\tilde {\mathcal{O}}_\kappa\) is a homeomorphism onto an open subset of \(\R^n\), with
the requirement that, given any non-empty overlap
\({\mathcal{O}}_\kappa\cap {\mathcal{O}}_{\kappa'}\),
the transition function \(\kappa'\circ\kappa^{-1}:
\kappa({\mathcal{O}}_\kappa\cap {\mathcal{O}}_{\kappa'})\to
\kappa'({\mathcal{O}}_\kappa\cap {\mathcal{O}}_{\kappa'}) \) is \(C^\infty \). 
Each chart gives a local coordinate system
\(
(x^1,\dots,x^n)
\) with \(x^j=\kappa^j(x)\forall x\in\mathcal{O}_\kappa\).
In such a coordinate system, the metric \(g\) is represented by a
matrix-valued function
\(
x\mapsto g_{ij}(x).
\)
When we say
  \((M,g)\) is a smooth Riemannian manifold it is to be understood both that \(M\) is \(C^\infty\) and that
  there is an atlas of coordinate charts in which \(g\) is represented by {\em smooth} matrix valued functions
  \(x\mapsto g_{ij}(x)\)
  (and without loss of generality these charts may be 
  taken to be harmonic by \cite{MR0644518}*{Theorem 4.5}).
In a local system of coordinates the Riemannian volume form is
\(
d\mu_g
=
\sqrt{\det g}\,
dx^1\wedge\cdots\wedge dx^n
=\sqrt{\det g}dx\). Let $\Omega\subset M$ be an open set with \(C^1\) boundary
\(
\partial\Omega.
\)
Let
\(
\nu=\nu^\alpha\frac{\partial}{\partial x^\alpha}
\)
denote the outward-pointing unit normal vector field along
$\partial\Omega$.
The induced hypersurface volume form on $\partial\Omega$ is defined by
\(
d\sigma_g
=
\iota_\nu(d\mu_g).
\)
Substituting the local expression for $d\mu_g$ gives
\[
d\sigma_g
=
\iota_\nu\left(
\sqrt{\det(g_{ij})}\,
dx^1\wedge\cdots\wedge dx^n
\right)\,,
\]
and using the contraction identity
\[
\iota_{\partial_{x^\alpha}}
(dx^1\wedge\cdots\wedge dx^n)
=
(-1)^{\alpha-1}
dx^1\wedge\cdots\wedge
\widehat{dx^\alpha}
\wedge\cdots\wedge dx^n,
\]
(where the \(\hat{\;\cdot\;}\) indicates omission), we obtain
\[
d\sigma_g
=
\sqrt{\det(g_{ij})}
\sum_{\alpha=1}^n
(-1)^{\alpha-1}\nu^\alpha
\,dx^1\wedge\cdots\wedge
\widehat{dx^\alpha}
\wedge\cdots\wedge dx^n.
\]
If
\(
X=X^\alpha\frac{\partial}{\partial x^\alpha}
\)
is a \(C^1\) vector field, then
\[
\operatorname{div}_g X
=
\frac{1}{\sqrt{\det g}}
\frac{\partial}{\partial x^\alpha}
\left(
\sqrt{\det g}\,X^\alpha
\right).
\]
The {\em divergence theorem} (\cite{MR2304165}*{\S\,I.3.1}]) states that if \(\Omega\) has a \(C^1\) boundary
and \(X\) is a \(C^1\) vector field
\begin{equation}\label{divth}
\int_\Omega \operatorname{div}_g(X)\,d\mu_g
=
\int_{\partial\Omega}
g(X,\nu)\,d\sigma_g.
\end{equation}

We write \(d_g\) for the Riemannian distance function; recall that it is continuously differentiable on the complement of the cut-locus and everywhere Lipshitz continuous with constant \(\leq 1\). In geodesic normal coordinates it has
the expression \eqref{defr}. The
ball of radius \(R\) in \(M\) is \(B_y(R)=\{x:d_g(x,y)<R\}\).

The {\em injectivity radius} \(r_{inj}(y)\) at \(y\) is the largest number such that the exponential map
restricted to the ball of tangent vectors \(v\in T_yM\) of length \(<r_{inj}(y)\) is a local diffeomorphism;
the {\em convexity radius} \(r_{conv}(y)\leq r_{inj}(y)\) at \(y\) is the radius of the largest ball in \(M\)
centered at \(y\) which is such that any two points in it are joined by a unique geodesic, and this geodesic is minimizing.
Taking the infimum over \(y\) gives the global versions of these,
\(r_{inj}=\inf_{y\in M} r_{inj}(y)\) etc.

\subsection{Weighted Sobolev spaces and the Laplacian}
We now define weighted Sobolev spaces on an asymptotically Euclidean manifold \((M,g)\). To be precise, we'll
assume the manifold has a description as follows.
\begin{enumerate}\item[(i)]
There is a \(C^\infty\) atlas of charts
\[\{\kappa,{\mathcal{O}}_\kappa,\tilde {\mathcal{O}}_\kappa\}_{\kappa\in\scriptscriptstyle{\mathcal{F}\cup\{\infty\}}} \]
where \(\mathcal{F}\) is a finite index set and,
for each \(\kappa\in\mathcal{F}\), \(\tilde {\mathcal{O}}_\kappa\) is a bounded open subset of \(\R^n\) and
\(M\subset \cup_{\kappa\in\scriptscriptstyle{\mathcal{F}\cup\{\infty\}}} {\mathcal{O}}_\kappa\).
There is a compact set \(K\subset M\) such that \[K\subset \cup_{\kappa\in\scriptscriptstyle{\mathcal{F}}} {\mathcal{O}}_\kappa\,,\qquad
  {\mathcal{O}}_\infty=M\backslash K\]
and the chart
\begin{equation}\label{ncend}
  \kappa_\infty:M\backslash K\to\{x\in\R^n:|x|>R_\infty\}
\end{equation}
defines {\em asymptotic coordinates} on the end
\(M\backslash K\).
\item[(ii)]
There is a subordinate \(C^\infty\) partition of unity
\(\{b_\kappa\}_{\kappa\in\scriptscriptstyle{\mathcal{F}\cup\{\infty\}}}\) with \(\hbox{supp\,} b_\kappa\) a compact subset of \(\mathcal{O}_\kappa\)
for \(\kappa\in\mathcal{F}\) and \(\hbox{supp\,}b_\infty=\kappa_\infty^{-1}(\{x\in\R^n:|x|\geq R_\infty'\})\) for some
\(R_\infty'>R_\infty\).
\end{enumerate}
Now for Sobolev spaces, the local spaces \(W^{k,q}_{loc}(M)\) can be defined unambiguously in terms of local coordinate systems.
For global spaces, the weighted norms can be defined in terms of the charts and subordinate partition of unity just described.
The condition \(\sum_{\kappa\in\scriptscriptstyle{\mathcal{F}\cup\{\infty\}}} b_{\kappa}=1\) everywhere means any real-valued function can be decomposed
\(u=\sum u_\kappa\) with \(u_\kappa=b_\kappa u\), and so the collection of functions
\(\{\tilde u_\kappa\}_{\kappa\in\scriptscriptstyle{\mathcal{F}\cup\{\infty\}}}\),
where \(\tilde u_\kappa=u_\kappa\circ \kappa^{-1}:\R^n\to\R\),
determine \(u:M\to\R\). (The functions \(\tilde u_\kappa\) are of course extended by zero outside \(\tilde{\mathcal{O}}_\kappa\).)
The weighted Sobolev norm determining the space
\(W^{k,q}_{-\tau}(M)\) is then defined by
\begin{equation}\label{wsob}
  \|u\|_{W^{k,q}_{-\tau}}\define \sum\nolimits_{\kappa\in\scriptscriptstyle{\mathcal{F}}}
    \|\tilde u_\kappa\|_{W^{k,q}}
  +
 \left(\sum_{|\alpha|\leq k}\int_{|x|>R_\infty}   |\partial^{\alpha} \tilde u_{\scriptscriptstyle{\infty}}(x)|^q\sigma(x)^{q{\tau+q|\alpha|}-n}\,dx\,\right)^{\frac{1}{q}}
  \end{equation}
  where \(\sigma(x)=\sqrt{1+|x|^2}\). This definition extends to tensor fields and vector valued functions in the obvious way.
\begin{remark}
  This definition
  does depend explicitly upon the atlas  and partition of unity. In the compact case different choices give rise to equivalent norms, but this is
  not necessarily the case in the presence of a noncompact end \eqref{ncend}. Nevertheless, if two such asymptotic structures at \(\infty\) exist for which
  \((M,g)\) is asymptotically Euclidean (in the sense of the definition which follows),
  then a type of uniqueness does hold, \cite{bartnik}. This is not crucial for the present paper, in which it is only required that one such
  structure does exist to draw conclusions about stability of matter.
  \end{remark}
\begin{definition}\label{bartnikdef}
Let \(q>n\) and \(\tau>0\). An
n-dimensional
\(W^{1,q}_{-\tau}\)-asymptotically Euclidean manifold is a \(C^\infty\)
manifold with an atlas and subordinate partition of unity as described above,
and in addition a Riemannian metric \(g\in W^{1,q}_{loc}(M)\) with the additional
property that there exists \(\lambda>1\) such that in the corresponding ``asymptotic''
coordinates on \(M\backslash K\) 
\begin{align}\label{asympe}\begin{split}
    \lambda^{-1}|\xi|^2\leq &g_{ij}(x)\xi^i\xi^j\leq \lambda |\xi|^2
    \qquad\forall |x| > R_\infty\,,\qquad\hbox{and}\\
    \sum_{i,j}\sum_{|\alpha|\leq 1}\int_{|x|> R_\infty}
    |\partial^{\alpha}(&g_{ij}(x)-\delta_{ij})|^q\sigma(x)^{\tau+q|\alpha|-n}\,dx\,<\infty\,,
  \end{split}
\end{align}
where the sum in the latter inequality is over multi-indices \(\alpha\)
of order \(\leq 1\) and all \(i,j\) in \(\{1,2,\dots n\}\).
\end{definition}

Thus the condition in \eqref{asympe} says that \(g -e\in W^{1,q}_{-\tau}\) where \(e\) is a smooth bounded
\(\begin{pmatrix}0\\2\end{pmatrix}\)
tensor field which extends
the Kronecker delta \(\delta_{ij}\) from \(\kappa_\infty(M\backslash K)\) to all of \(M\).
Since \(K\) is compact
  we can assume, by increasing \(\lambda\) as necessary, that the first condition of \eqref{asympe} holds
  everywhere, i.e., in all coordinate systems we'll work with an atlas of charts in which
  \begin{equation}\label{asympa}
\lambda^{-1}|\xi|^2\leq g_{ij}(x)\xi^i\xi^j\leq\lambda|\xi|^2\,,\qquad\forall x\in M.
    \end{equation}

  The weighted Sobolev inequality for dimension \(n=3\)
  says that \((M,g)\) and \(q\) are as above, and if \(u\in W^{2,p}_{-\tau}\) for \(q\geq p>3\), then \(u\) agrees
  (as a distribution, strictly speaking) with a \(C^1\) function, also written \(u\), and
  \begin{equation}\label{sobw}
    \max_{M}\left(\sigma^\tau |u|+\sigma^{1+\tau}|\partial u|\right)
    \leq C_S\|u\|_{W^{2,p}_{-\tau}}\,.
    \end{equation}
    We will also need
    the fact that in this situation
    \begin{equation}\label{biso}
-\Delta_g : W_{-\tau}^{2,p} \to W_{-\tau-2}^{0,p} = L^p_{-\tau-2}
      \end{equation}
is an isomorphism for $\tau \in (0,1)$, see \cite{bartnik}*{Prop 2.2}.

\subsection{Hadamard Expansions for Green Functions}

In this note only the three-dimensional case is required.
Let $(M,g)$ be a smooth $3$-manifold with a metric \(g\) whose coefficients
\(x\mapsto g_{ij}(x)\) lie in \(C^\infty\) in a given system of coordinates. Then the local
regularity for uniformly elliptic equations from \cite{MR0737190} applies to the equation \eqref{115},
and so any solution lies in \(C^\infty(M\backslash \{y\})\).
Choose geodesic normal coordinates centred at $y$, in which $y$ is represented by $x^j = 0\forall j$, and the condition
\begin{equation}\label{gnc}
g_{ij}(x)x^j = g_{ij}(0)x^j
\end{equation}
holds throughout. Since we will later multiply by a cutoff function supported in the domain of a normal coordinate system, we may assume the condition holds everywhere. With these coordinates we then have
\begin{equation}\label{defr}
d_g(x,y)=r(x)
  =\left(g_{ij}(0)x^ix^j\right)^{\frac12}
  =\left(g_{ij}(x)x^ix^j\right)^{\frac12}\,.
  \end{equation}
We summarize the Hadamard parametrix expansions following \cite{MR2304165}*{\S\,III.17.4}.
The functions $F_0, F_1, F_2, \dots$ which support the Hadamard approximation are given by
\begin{equation}
F_{\nu}(x) = (-1)^{\nu} \frac{{\nu}!}{(2{\nu})!} \frac{1}{4\pi} r(x)^{2{\nu}-2},
\end{equation}
where $r(x)^2= g_{ij}(x)x^i x^j$ is the radial coordinate (taken positive).  
In particular,
\[
F_0 = \frac{1}{4\pi r}, \qquad 
F_1 = -\frac{r}{8\pi }, \qquad
F_2 = \frac{r^3}{48\pi }.
\]
We need the relations:
\begin{align}\label{defrel}\begin{split}
{\nu} = 0 &: \quad -\partial_j(g^{jk}\partial_k F_0) = \sqrt{\det g^{-1}} \delta_0\,, \\
{\nu} > 1 &: \quad -\partial_j(g^{jk}\partial_kF_{\nu}) = {\nu} F_{{\nu}-1}\,, \quad\hbox{and}\\
&\;\; \quad -2 g^{ik} \partial_k F_{\nu} = x^i F_{{\nu}-1}\,.
\end{split}\end{align}
Define, for ${\nu} = 0,1,2,\dots$,
\begin{equation}\label{defln}
L_{\nu} = x^i \partial_i +\nu- \frac12 g_{ij} x^ib^j\,.
\end{equation}
Then calculation gives, for $u_\nu \in C^2$,
\begin{align}\label{imp}\begin{split}
-\Delta_g \sum_{\nu=0}^N u_\nu F_\nu &= u_0 \sqrt{\det g^{-1}} \delta_0 + (2\pi r^3)^{-1}L_0 u_0 \\
&\quad + \sum_{{\nu}=1}^N (L_{\nu} u_{\nu} - \Delta_g u_{{\nu}-1}) F_{{\nu}-1} - (\Delta_g\, u_N) F_N.
\end{split}\end{align}
The Hadamard expansion to order $N$ is then
\(
\sum_{{\nu}=0}^N U_{\nu} F_{\nu},
\)
where $U_0(0)=1$, and the $\{U_{\nu}\}$ solve the transport equations
\begin{align}
L_0 U_0 &= 0, \\
L_0 U_{\nu} - \Delta_g U_{{\nu}-1} &= 0, \qquad {\nu} = 1,\dots,N.
\end{align}
This leads to
\begin{equation}
 -\Delta_g \left( G_{0}(x) - \sum_{{\nu}=0}^NU_{\nu} F_{\nu} \right)
= -(\Delta_g U_N) F_N,
\end{equation}
and so gives an asymptotic expansion for $G_{0}(x)$ as $r = d(x,{0}) \to 0$.
To obtain a global statement, introduce a smooth cutoff $\beta \in C^\infty(\R)$ with
\[
\beta(t) = 1 \quad \text{for } t\leq \rho^2_1<r^2_{*}\,,
\]
\[
\beta(t) = 0 \quad \text{for } t \geq \rho_2^2\,,
\]
with \(\rho_1<\rho_2\leq r_{*}\).
Inserting as cutoff \(\tilde\beta(x)=\beta(r(x)^2)\) into the expansion leads to
\begin{equation}\label{coa}
-\Delta_g \left(G_{0}(x) - \tilde\beta \sum_{{\nu}=0}^N U_{\nu} F_{\nu} \right)
= -\Delta_g (\tilde\beta U_N) F_N + \mathcal{E}_N\,,
\end{equation}
where the cutoff error is
\begin{equation}
  \mathcal{E}_N \define\sum_{{\nu}=1}^N \left( U_\nu x^j\partial_j\tilde\beta) - (\Delta_g \tilde\beta) U_{{\nu}-1}
-2g^{ij}\partial_i\tilde\beta\partial_j U_{\nu-1}
  \right)\,.
\end{equation}
This is supported in $\{x : d(x,{0}) < \rho_2\}$. Since $\rho_2<r_{*}$ is smaller than the injectivity radius, the use of the normal coordinate condition above is justified, and we can then make a global estimate by means of \eqref{biso}.
Choosing $N=1$, the transport equations are solved to give
\begin{equation}\label{ste}\begin{split}
U_0(x,{0}) &= \left( \frac{\det g({0})}{\det g(x)} \right)^{1/4}\,,
\\
U_1(x,{0}) &= U_0(x) \int_0^1 \frac{\Delta_g\, U_0(sx)}{U_0(sx)} \, ds.
\end{split}\end{equation}
Observe that \(U_0(x,{0})\), which involves $(\det g(x))^{-\frac{1}{4}}$, appears directly in \(\mathcal{E}_N\) with \(N=1\) differentiated once,
and indirectly differentiated twice through \(U_1\).
But also the right hand side of \eqref{coa} indirectly contains \(U_0\)
differentiated four times via the appearance of 
\(\Delta_g (\tilde\beta U_1) F_1 \). Now
given a square matrix \(A\), its determinant is
a polynomial in the matrix elements $A_{ij}$, while (by the Cramer rule) its inverse is a rational
function of the matrix elements (with the determinant in the denominator). It is here that the
condition \eqref{IB2} is needed: it implies that \(x\mapsto \det g(x)\) is a {\em strictly positive}
polynomial function of the \(g_{ij}(x)\), and as such the Moser inequalities
imply bounds in \(B_0(r_*)=\{x:r(x)<r_*\}\)
\begin{equation}\label{bmieq}
  \|U_0\|_{W^{4,q}(B_0(r_*))}
+ \|U_1\|_{W^{2,q}(B_0(r_*))}
  \leq c_1(L,r_{*}, Q)\,.
  \end{equation}
Hence applying \eqref{biso} to \eqref{coa} (and taking \(N=1\)) we can bound 
  \[
\|\left(G_{0}(x) - \tilde\beta \sum_{{\nu}=0}^1 U_{\nu} F_{\nu} \right)\|_{ W_{-\tau}^{2,q}}\,\leq\, c_2(L,r_{*}, Q)\,.
    \]
For \(q>3\) this implies the Hadamard asymptotics with \(C^1\) bounds stated in (I2).
%
Now let \(\{\xi^a(r,\theta^1,\theta^2)\}_{a=1}^3\) be given by the usual three dimensional spherical coordinate formulae
\[
\xi^1=r\cos{\theta^1}\sin\theta^2\,,\;\xi^2=r\sin{\theta^1}\sin\theta^2\quad\hbox{and}\;\xi^3=r\cos\theta^2
\]
and let
\(T^i_a\) satisfy \(g_{ij}(0) T^i_aT^j_b=\delta_{ab}\), then putting \(x^i=T^i_a\xi^a\) gives
\begin{equation}\label{gsc}
g_{ij}(x)\,dx^idx^j=dr^2+g_{\alpha\beta}(r,\theta)d\theta^\alpha d\theta^\beta\,,\qquad\alpha,\beta\in\{1,2\}\,,
\end{equation}
for \(r(x)<r_{inj}\). [The \(g_{r\alpha}\) terms in the metric vanish by the Gauss lemma. The condition
\eqref{IB1} implies that for \(r<r_*\) the spherical part of the metric \(g_{\alpha\beta}\) is bounded above and below in terms of \(k_1,k_2\)
by the Jacobi comparison principle, see \cite{MR3469435}*{Corollary 2.4}.]
Thus \((r,\theta^1,{\theta^2})\) provide a system of {\em geodesic spherical coordinates} in a neighbourhood
\(\{x:d_g(x,y)<r_{*}\}\) of \(y\). and usefully the formula
\[
\nabla_x G(x,y)
=-\frac{\nabla_x r}{4\pi r^2}+O({1})\,,
\]
implies that near to \(y\) the Green function \(G_y(x)\) is
strictly monotonically decreasing in \(r\) (i.e. in radial directions),
and hence that on large \(\alpha\) level sets
\(
\Omega_{\alpha,y}=\{x:G(x,y)=\alpha\}
\)
it is possible to solve locally for \(r\) in terms of \(\alpha,\theta^1, {\theta^2}\), thus expressing
the level set as a \(C^1\) graph over the sphere, i.e. the level sets
are smooth perturbations of Euclidean spheres of radius
$(4\pi\alpha)^{-1}$.

\noindent
{\bf Acknowledgements} Thanks to Robin Reuvers and Philip Solovej for discussions on the stability of matter in a related context.

\small
\baselineskip=13pt


\begin{bibdiv}
\begin{biblist}

  \bib{MR1074481}{article}{
   author={Anderson, Michael T.},
   title={Convergence and rigidity of manifolds under Ricci curvature
   bounds},
   journal={Invent. Math.},
   volume={102},
   date={1990},
   number={2},
   pages={429--445},
   issn={0020-9910},
   review={\MR{1074481}},
   doi={10.1007/BF01233434},
}
    
\bib{bartnik}{article}{
   author={Bartnik, Robert},
   title={The mass of an asymptotically flat manifold},
   journal={Comm. Pure Appl. Math.},
   volume={39},
   date={1986},
   number={5},
   pages={661--693},
   issn={0010-3640},
   review={\MR{0849427}},
   doi={10.1002/cpa.3160390505},
 }

 \bib{MR0586803}{article}{
   author={Baxter, J. R.},
   title={Inequalities for potentials of particle systems},
   journal={Illinois J. Math.},
   volume={24},
   date={1980},
   number={4},
   pages={645--652},
   issn={0019-2082},
   review={\MR{0586803}},
}

 \bib{MR0644518}{article}{
   author={DeTurck, Dennis M.},
   author={Kazdan, Jerry L.},
   title={Some regularity theorems in Riemannian geometry},
   journal={Ann. Sci. \'Ecole Norm. Sup. (4)},
   volume={14},
   date={1981},
   number={3},
   pages={249--260},
   issn={0012-9593},
   review={\MR{0644518}},
}

\bib{MR2787767}{article}{
   author={Enciso, Alberto},
   title={Coulomb systems on Riemannian manifolds and stability of matter},
   journal={Ann. Henri Poincar\'e},
   volume={12},
   date={2011},
   number={4},
   pages={723--741},
   issn={1424-0637},
   doi={10.1007/s00023-011-0084-5},
 }

 \bib{MR1625845}{book}{
   author={Evans, Lawrence C.},
   title={Partial differential equations},
   series={Graduate Studies in Mathematics},
   volume={19},
   publisher={American Mathematical Society, Providence, RI},
   date={1998},
   pages={xviii+662},
   isbn={0-8218-0772-2},
   review={\MR{1625845}},
   doi={10.1090/gsm/019},
 }

 \bib{MR1681462}{book}{
   author={Folland, Gerald B.},
   title={Real analysis},
   series={Pure and Applied Mathematics (New York)},
   edition={2},
   note={Modern techniques and their applications;
   A Wiley-Interscience Publication},
   publisher={John Wiley \& Sons, Inc., New York},
   date={1999},
   pages={xvi+386},
   isbn={0-471-31716-0},
   review={\MR{1681462}},
}

\bib{MR0737190}{book}{
   author={Gilbarg, David},
   author={Trudinger, Neil S.},
   title={Elliptic partial differential equations of second order},
   series={Grundlehren der mathematischen Wissenschaften [Fundamental
   Principles of Mathematical Sciences]},
   volume={224},
   edition={2},
   publisher={Springer-Verlag, Berlin},
   date={1983},
   pages={xiii+513},
   isbn={3-540-13025-X},
   review={\MR{0737190}},
   doi={10.1007/978-3-642-61798-0},
}
 
 \bib{MR2304165}{book}{
   author={H\"ormander, Lars},
   title={The analysis of linear partial differential operators. I-IV},
   series={Classics in Mathematics},
   publisher={Springer, Berlin},
   isbn={978-3-540-49937-4},
   review={\MR{2304165}},
   doi={10.1007/978-3-540-49938-1},
}

\bib{MR1393941}{book}{
   author={Kobayashi, Shoshichi},
   author={Nomizu, Katsumi},
   title={Foundations of differential geometry. Vol. II},
   series={Wiley Classics Library},
   note={Reprint of the 1969 original;
   A Wiley-Interscience Publication},
   publisher={John Wiley \& Sons, Inc., New York},
   date={1996},
   pages={xvi+468},
   isbn={0-471-15732-5},
   review={\MR{1393941}},
}


\bib{MR2583992}{book}{
   author={Lieb, Elliott H.},
   author={Seiringer, Robert},
   title={The stability of matter in quantum mechanics},
   publisher={Cambridge University Press, Cambridge},
   date={2010},
   pages={xvi+293},
   isbn={978-0-521-19118-0},
   review={\MR{2583992}},
 }

 \bib{MR3469435}{book}{
   author={Petersen, Peter},
   title={Riemannian geometry},
   series={Graduate Texts in Mathematics},
   volume={171},
   edition={3},
   publisher={Springer},
   date={2016},
   pages={xviii+499},
   isbn={978-3-319-26652-7},
   isbn={978-3-319-26654-1},
   doi={10.1007/978-3-319-26654-1},
}

\bib{solorev}{article}{
   author={Solovej, J. P.},
   title={Stability of matter},
     conference={
      title={Encyclopedia of mathematical physics. Vol. 2. Quantum
      mechanics, quantization and quantum computation},
   },
   book={
      publisher={Academic Press, Amsterdam},
   },
   isbn={978-0-443-29953-7},
   isbn={978-0-323-95703-8},
   date={[2025] \copyright 2025},
   pages={120--126},
   review={\MR{4975271}},
 }
\end{biblist}
\end{bibdiv}
\end{document}